\documentclass[journal]{IEEEtran}
\usepackage{amsmath,nccmath}
\usepackage{amssymb}
\usepackage{lipsum}
\usepackage{graphicx}
\usepackage{esint}
\usepackage{amsmath,amssymb,amsthm}
\usepackage{thmtools,thm-restate}
\usepackage{hyperref}
\usepackage{cleveref}
\usepackage{ifpdf}
\usepackage{cite}
\usepackage{color}
\usepackage{placeins}
\usepackage{float}
\usepackage{tabularx}
\usepackage{colortbl}
\usepackage{pgfplots}
\usepackage{tikz}
\usepackage{tikzscale}
\pgfplotsset{compat=newest}
\usetikzlibrary{plotmarks}
\usetikzlibrary{arrows.meta}
\usepgfplotslibrary{patchplots}
\usepackage{grffile}
\pgfplotsset{plot coordinates/math parser=false}
\newlength\figureheight
\newlength\figurewidth
\usepackage{pgfgantt}
\usepackage{pdflscape}

\begin{document}

\title{Pincer-Based vs. Same-Direction Strategies of Search for Smart Evaders by Swarms of Agents}


\author{Roee M. Francos
        and Alfred M. Bruckstein
\thanks{Roee M. Francos and Alfred M. Bruckstein are with the Faculty
of Computer Science, Technion- Israel Institute of Technology, Haifa, Israel, 320003, Emails:(roee.francos@cs.technion.ac.il,
  alfred.bruckstein@cs.technion.ac.il).}
}

\maketitle
\begin{abstract}
Suppose in a given planar region, there are smart mobile evaders and we want to detect them using sweeping agents. We assume that the agents have line sensors of equal length. We propose procedures for designing cooperative sweeping processes that ensure successful completion of the task, thereby deriving conditions on the sweeping speed of the agents and their paths. Successful completion of the task means that evaders with a known limit on their speed cannot escape the sweeping agents. A simpler task for the sweeping swarm is the confinement of the evaders to their initial domain. The feasibility of completing these tasks depends on geometric and dynamic constraints that impose a lower bound on the speed the sweeping agent must have. This critical speed is derived to ensure the satisfaction of the confinement task. Increasing the speed above the lower bound enables the agents to complete the search task as well. We present a quantitative and qualitative comparison analysis between the total search time of same-direction sweep processes and pincer-movement search strategies. We evaluate the different strategies by using two metrics,  total search time and the minimal critical speed required for a successful search. We compare two types of pincer-movement search processes, circular and spiral, with their same-direction counterparts, for any even number of sweeping agents. We prove that pincer based strategies provide superior results in all practical scenarios and that the spiral pincer sweep process allows detection of all evaders while sweeping at nearly theoretically optimal speeds.
\end{abstract}


\IEEEpeerreviewmaketitle

\section{Introduction}
\IEEEPARstart{T}{he}
aim of this work is to provide an efficient "must-win" search policy for a swarm of $n$ sweeping agents that must guarantee detection of an unknown number of smart evaders initially residing inside a given circular region of radius $R_0$ while minimizing the search time. The evaders move and try to escape the initial region at a maximal speed of $V_T$, known to the sweepers. All sweepers move at a speed $V_s > V_T$ and detect the evaders using linear sensors of length $2r$. Each "must-win" policy requires a minimal speed that depends on the trajectory of the sweepers. Finding an efficient algorithm requires that, throughout the sweep, the footprint of the sweepers’ sensors maximally overlaps the evader region (the region where evaders may possibly be). This work develops two "must-win"  same-direction search strategies, circular and spiral, for a swarm consisting of an even number of searchers that sweep the evader region until all evaders are detected. Afterwards, a comparison between the developed same-direction and pincer-based search strategies developed in \cite{francos2021tro} is performed.
\subsection{Overview of Related Research}

Several interesting search problems originated in the second world war due to the need to design patrol strategies for aircraft aiming to detect ships or submarines in the English channel, see \cite{koopman1980search}. The problem of patrolling a corridor using multi agent sweeping systems in order to ensure the detection and interception of smart targets was also investigated in \cite{vincent2004framework} and provably optimal strategies were provided in \cite{altshuler2008efficient}.

In \cite{bressan2008blocking,bressan2012optimal}, Bressan et al. investigate optimal strategies for the construction of barriers in real time aiming at containing and confining the spread of fire from a given initial area of the plane. The authors are interested in determining a minimal possible barrier construction speed that enables the confinement of the fire, and on determining optimality conditions for confinement strategies. 

A non-escape search procedure for evaders that are originally located in a convex region of the plane and may move out of it is investigated in \cite{tang2006non}, and a cooperative progressing spiral-in algorithm performed by several agents with disk shaped sensors in a leader-follower formation is proposed. In \cite{mcgee2006guaranteed}, McGee et al. also investigate a search problem for smart targets that do not have any maneuverability restrictions except for an upper limit on their speed. The sensor the agents are equipped with detects targets within a disk shaped area around the searcher location. Search patterns consisting of spiral and linear sections are considered. In \cite{hew2015linear}, Hew proposes searching for smart evaders using concentric arc trajectories with agents sensors similar to \cite{mcgee2006guaranteed}. Such a search is aimed at detecting submarines in a channel or in a half plane.

Another set of related problems are pursuit-evasion games, where the pursuers' objective is to detect evaders and the evaders objective is to avoid the pursuers. Pursuit-evasion games include combinations of single and multiple evaders and pursuers scenarios. In this context several works considered the problem of defending a region from the entrance of intruders. In \cite{fisac2015reach,chen2016multiplayer}, such problems are investigated under the name of ``reach-avoid games". These types of problems were also addressed in the context of perimeter defense games by Shishika et al. in \cite{shishika2018local,shishika2020cooperative}, with a focus on utilizing cooperation between pursuers to improve the defense strategy. In \cite{shishika2018local}, implicit cooperation between pairs of defenders that move in a ``pincer movement" is performed in order to intercept intruders before they enter a convex region in the plane. In \cite{makkapati2019optimal}, pursuit–evasion problems involving multiple pursuers and multiple evaders (MPME) are studied. Pursuers and evaders are all assumed to be identical, and pursuers follow either a constant bearing or a pure pursuit strategy. The problem is simplified by adopting a dynamic divide and conquer approach, where at every time instant each evader is assigned to a set of pursuers based on the instantaneous positions of all the players. The original MPME problem is decomposed to a sequence of simpler multiple pursuers single evader (MPSE) problems by classifying if pursuer is relevant or redundant for each evader by using Apollonius circles. Only the relevant pursuers participate in the MPSE pursuit of each evader. 

In \cite{francos2019search}, the confinement and cleaning tasks for a line formation of agents or alternatively for a single agent with a linear sensor are analyzed. In \cite{francos2021tro}, teams of agents perform pincer sweep search strategies with linear sensors.

Recent surveys on pursuit evasion problems are \cite{ chung2011search,kumkov2017zero,weintraub2020introduction}. In \cite{chung2011search}, a taxonomy of search problems is presented. The paper highlights algorithms and results arising from different assumptions on searchers, evaders and environments and discusses potential field applications for these approaches. The authors focus on a number of pursuit-evasion games that are directly connected to robotics and not on differential games which are the focus of the other cited surveys. \cite{kumkov2017zero} presents a survey on pursuit problems with $1$ pursuer versus $2$ evaders or $2$ pursuers versus $1$ evader are formulated as a dynamic game and solved with general methods of zero-sum differential games.
In \cite{weintraub2020introduction}, the authors present a recent survey on pursuit-evasion differential games and classify the papers according to the numbers of participating players: single-pursuer single-evader (SPSE), MPSE, one- pursuer multiple-evaders (SPME) and MPME.

\subsection{Contributions} 
 
In this paper, we provide several theoretical and experimental contributions to multi-agent search and coordinated motion planning literature. We propose search protocols that guarantee detection of all smart evaders that are initially located in given circular region from which they may move out of in order to escape the pursuing sweeping agents. A detailed theoretical analysis of trajectories, critical speeds and search times for same-direction sweep protocols performed by a swarm of $n$ cooperative agents are developed in order to quantitatively compare these methods to the pincer-based protocols described in \cite{francos2021tro}.  
 
\begin{itemize}
 \item We propose two types of same-direction sweep protocols:
        \begin{itemize}
        \item  Same-direction circular sweep pincer sweep strategy 
        \item  Same-direction spiral sweep pincer sweep strategy 
        \end{itemize}
\item We prove that for both same-direction sweep protocol types, the corresponding pincer-based protocols yield a lower critical speed.
\item We show that the circular pincer-based sweep protocol always results in shorter sweep times compared to its same-direction counterpart.
\item Results show that as the number of sweepers increases, circular pincer-based protocols require a smaller critical speed even when compared to spiral same-direction protocols. This result indicates that although implementing pincer-based circular search protocols requires sweepers with more basic capabilities compared to spiral protocols, the cooperation between the sweepers considerably improves the overall performance of the sweeper team. 
\item We experimentally show that for all choices of search parameters the spiral pincer-based protocol is favourable compared to its same-direction counterpart and that it results in critical speeds that approach the theoretical lower bound.

\end{itemize}

\section{Same-direction Versus Pincer-based Sweeps}

This paper considers a scenario in which a multi-agent swarm of identical agents search for mobile targets or evaders that are to be detected. The information the agents perceive only comes from their own sensors, and all evaders that intersect a sweeper's field of view are detected. We assume that all agents have a linear sensor of length $2r$. The evaders are initially located in a disk shaped region of radius $R_0$. There can be many evaders, and we consider the domain to be continuous, meaning that evaders can be located at any point in the interior of the circular region at the beginning of the search process. All sweepers move with a speed of $V_s$ (measured at the center of the linear sensor). By assumption the evaders move at a maximal speed of $V_T$, without any maneuverability restrictions. The sweeper swarm's objective is to "clean" or to detect all evaders that can move freely in all directions from their initial locations in the circular region of radius $R_0$. 

Search time clearly depends on the type of sweeping movement the searching team employs. Detection of evaders is based on deterministic and preprogrammed search protocols. We consider two types of search patterns, circular and spiral patterns. The desired result is that after each sweep around the region, the radius of the circle that bounds the evader region (for the circular sweep), or the actual radius of the evader region (for the spiral sweep), decreases by a strictly positive value. This guarantees complete cleaning of the evader region, by shrinking in finite time the possible area in which evaders can reside to zero. At the beginning of the circular search process we assume that only half the length of the agents' sensors is inside the evader region, i.e. a footprint of length $r$, while the other half is outside the region in order to catch evaders that may move outside the region while the search progresses. At the beginning of the spiral search process we assume that the entire length of the agents' sensors is inside the evader region, i.e. a footprint of length $2r$.

In the single agent search problem described in \cite{francos2019search}, we observed that there can be escape from point $P=(0,R_0)$ (shown in Fig. $1 \hspace{1mm} (a)$), when basing the searcher's speed only on a single traversal around the evader region. Therefore we had to increase the agent's critical speed to deal with this possible escape. Point $P$ is considered as the "most dangerous point", meaning that it has the maximum time to spread during sweeper movement, hence if evaders spreading from this point are detected, evaders escaping from all other points are detected as well. If we choose to distribute a multi-agent swarm equally along the boundary of the initial evader region, we would have the same problem of possible escape from the points adjacent to the starting locations of every sweeper.

In \cite{francos2021tro} we proposed an alternative method for multi-agent sweep processes such that pairs of sweeping agents move out in opposite directions along the boundary of the evader region and sweep in a pincer movement rather than having a convoy of sweepers all moving in the same-direction along the boundary. The proposed method is applicable for any even number of sweepers. The sweepers are initially positioned in pairs back to back. One sweeper in the pair moves counter clockwise while the other sweeper in the pair moves clockwise. Once the sweepers meet, i.e. their sensors are again superimposed at a meeting point, they switch the directions in which they move. This changing of directions occurs every time a sweeper bumps into another. Each sweeper is responsible for an angular sector of the evader region that is proportional to the number of participating agents in the search. Sweeping with a pincer-based search protocol eliminates the need to sweep additional areas to detect evaders from these additional "most dangerous points" since in pincer-based protocols the "most dangerous points" are now located at the tips of their sensors closest to the evader region's center.

\begin{figure}[ht]
\noindent \centering{}\includegraphics[width=1.36in,height =2.56in]{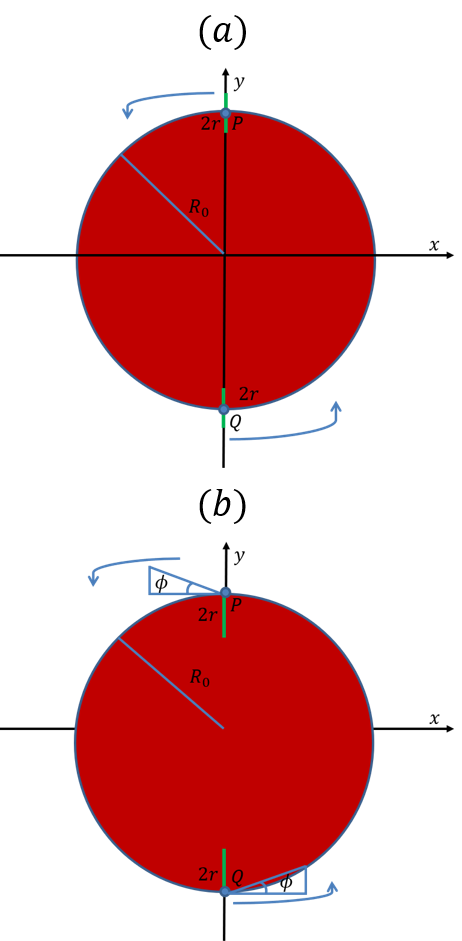} \caption{(a) - Initial placement of $2$ agents employing the same-direction circular sweep process.  (b) - Initial placement of $2$ agents employing the same-direction spiral sweep process. The sweepers sensors are shown in green. The angle $\phi$ is the angle between the tip of a sweeper's sensor and the normal of the evader region. $\phi$ is an angle that depends on the ratio between the sweeper and evader speeds.}
\label{Fig1Label}
\end{figure}

 .


\section{Same-direction Circular Sweep}
\subsection{Circular Sweep Time Calculation}
Previously in \cite{francos2019search} we tried to find the tightest lower bound of a searcher's speed by constructing a function of $2$ variables $f(t,V_s)$, by demanding that the furthest possible spread of the evader region is cleaned by the furthest tip of the sweeper's line sensor. A lesser requirement is to demand that by the time the most problematic point in the evader region, point $P$, spreads to a possible circle of radius of $r$ around point $P$, the sweeping swarm completes in addition to a sweep of $\frac{2\pi}{n}$ around the evader region an additional angular traversal that is proportional to traversing an arc of length $r$. This means that the agent travels an angle of $\frac{2\pi}{n} +\beta_0$ where $\beta_0$ is marked in Fig. $2(a)$. This assumption results in a simplified expression for the critical speed that bounds the previously found critical speed of  \cite{francos2019search} from above for all choices of geometric parameters.
\begin{figure}[ht]
\noindent \centering{}\includegraphics[width=2in,height =1.28in]{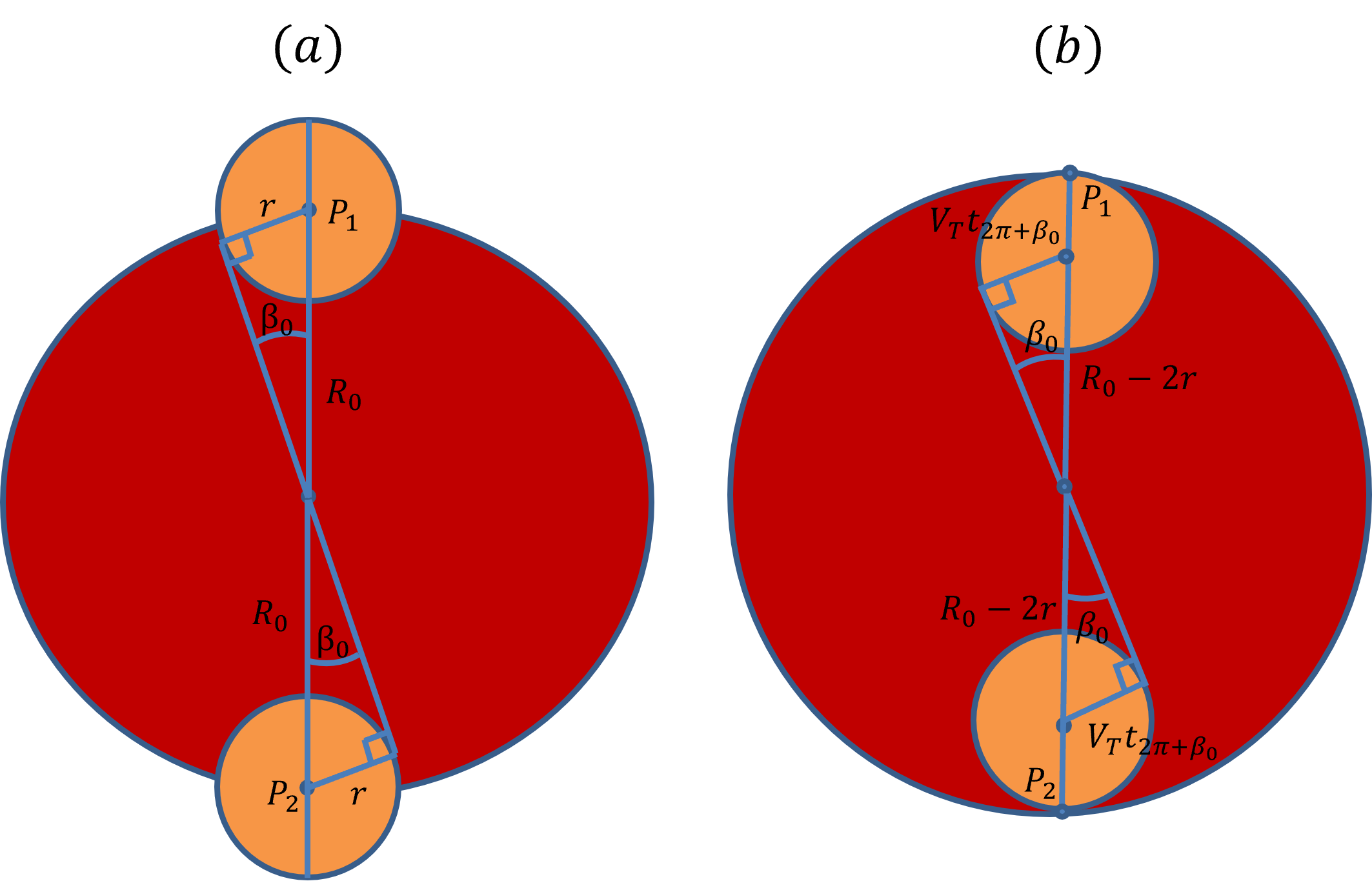} \caption{Geometric representation required for critical speed calculation. Red areas indicate locations where potential evaders may be located. The orange circles denote the spread of potential evaders around the most problematic points $P_1$ and $P_2$  during a traversal of $2\pi +\beta_0$ around the evader region. (a) - Same direction circular sweep. (b) - Same direction spiral sweep.}
\label{Fig2Label}
\end{figure}
We denote the time it takes the most problematic points to spread a distance of $r$ as $T_e$. These points are adjacent to the starting locations of the sweepers, and $2$ such points $P_1$ and $P_2$, exist in case the search is performed with $2$ sweepers, as shown in Fig. $2$. We have that $T_e = \frac{r}{{{V_T}}}$. We can see from Fig. $2$. that $\sin \beta_0  = \frac{r}{{R_0}}$, therefore $\beta_0  = \arcsin \frac{r}{{R_0}}$. The time it takes the sweeper to travel an angle of $\frac{2\pi}{n} +\beta_0$ is therefore given by ${T_s} = \frac{{\left( {\frac{2\pi}{n}  + \arcsin\left( \frac{r}{{R_0}}\right)} \right){R_0}}}{{{V_s}}}$. In order to guarantee no escape, we demand that $T_s \leq T_e$. Therefore, rearranging terms in the previous equation and plugging $T_e$ instead of $T_s$ yields,
\begin{equation}
{V_c} \ge \frac{{\left( {\frac{2\pi}{n}  + \arcsin\left( \frac{r}{{R_0}}\right)} \right){R_0}{V_T}}}{r}
\label{e24}
\end{equation}
The lower bound on a sweeper speed that ensures confinement is obtained when we have equality in (\ref{e24}). In future derivations we use the first order Taylor approximation for the arcsine function in (\ref{e24}), in order to enable the construction of analytical results for the sweep times of the evader region. Such an approximation is valid since in all practical scenarios the ratio between $\frac{r}{R_0}$ is sufficiently small. Applying this approximation to (\ref{e24}) allows us to define ${V_{{c_{circ}}}}$, the chosen critical speed, given by,
\begin{equation}
{V_{{c_{circ}}}} = \frac{{2\pi {R_0}{V_T}}}{{rn}} + {V_T}
\label{e25}
\end{equation}
In order for the sweeper swarm to advance inward toward the center of the evader region it must travel in a speed that is greater than the critical speed. We denote by $\Delta V>0$ the increment in the sweeping agents' speed that is above the critical speed. Each agent's speed ${V_s}$ is therefore given by the sum of the critical speed and $\Delta V$, namely
${V_s} = {V_{{c_{circ}}}} + \Delta V$. 
The total sweep times it takes the sweeper swarm to  reduce the evader region to a region bounded by a circle with a radius that is smaller or equal to $r$ is given by the sum of the circular motions and inward advancements that are performed after the completion of each circular sweep. The time it takes the sweepers to perform the circular sweeps is given by,
\begin{equation}
\begin{array}{l}
{T_{circular}} =  - \frac{{{R_0}\left( {{V_s} + {V_T}} \right)}}{{{V_s}{V_T}}} + \frac{{r\left( {{V_s} - {V_T}} \right)\left( {n\left( {{V_s} + {V_T}} \right) + 2\pi {V_T}{N_n}} \right)}}{{2\pi {V_T}^2{V_s}}} + \\ {\left( {1 + \frac{{2\pi {V_T}}}{{n\left( {{V_s} + {V_T}} \right)}}} \right)^{{N_n}}}\left( {{V_s} + {V_T}} \right)\left( {\frac{{2\pi {R_0}{V_T} - rn\left( {{V_s} - {V_T}} \right)}}{{2\pi {V_T}^2{V_s}}}} \right) \\ + \frac{{2\pi r}}{{n{V_s}}}
\end{array}    
\label{e26}
\end{equation}
The time it takes the sweepers to perform the inward advancement is given by,
\begin{equation}
\begin{array}{l}
{T_{in}} = \frac{{{R_0}}}{{{V_s}}} + {\left( {1 + \frac{{2\pi {V_T}}}{{{V_s} + {V_T}}}} \right)^{N - 1}}\left( {\frac{{2\pi {R_0}{V_T} - r\left( {{V_s} - {V_T}} \right)}}{{{V_s}\left( {{V_s} + {V_T}} \right)}}} \right)
\end{array}
\label{e27}
\end{equation}
The full analytical development is provided in Appendix $A$. 

\subsection{Same-direction Circular Sweep End-game}
In order to entirely clean the evader region the sweepers need to change the scanning method when the evader region is bounded by a circle of radius $r$. This is due to the fact that a smart evader that is very close to the center of the evader region can travel at a very high angular velocity compared to the angular velocity of the pursuing agents. This constraint is described by the following two equations, ${\omega _s} = \frac{{{V_s}}}{r}, {\omega _T} = \frac{{{V_T}}}{\varepsilon }$. The first describes the searcher's angular velocity and the second the evader's angular velocity. Since $\varepsilon$ can be arbitrarily small the evader can move just behind a sweeper's sensor and never be detected. Thus a slight modification to the sweep process needs to be applied in order to clean the entire evader region with the sweeper swarm that employs a circular scan. After completing sweep number $N_n -1$ the sweepers move toward the center of the evader region until the tip of the sweeper's sensors closest to the center of the evader region are placed at the center of the evader region. Following this motion the sweepers perform a circular sweep of radius $r$ around the center of the evader region. The time this last circular sweep takes is given by ${T_{last}} = \frac{{2\pi r}}{{{n V_s}}}$. Therefore, after the last circular scan the evader region is bounded by a circle of radius ${R_{last}}$, given by,
\begin{equation}
{R_{last}} = {T_{last}}{V_T} = \frac{{2\pi r V_T}}{{{n V_s}}}
\label{e28}
\end{equation}
In order to overcome the challenges in the circular search that were described we propose that after scan number $N_n + 1$ the sweeper swarm will travel to the right until cleaning the wavefront that propagates from the right portion of the remaining evader region and then travel to the left until cleaning the remaining evader region. A depiction of the scenario at the beginning of the end-game is presented in Fig. $3$. Theorem $1$ states the conditions for this demand to hold.

\begin{figure}[ht]
\noindent \centering{}\includegraphics[width=3.4in,height =1.7in]{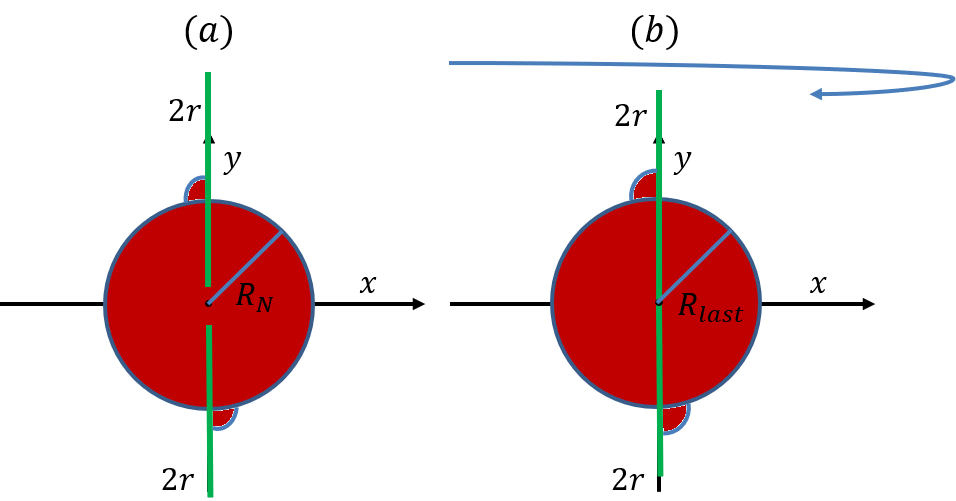} \caption{Depiction of the end-game steps for the same-direction circular sweep performed by $2$ sweepers. The sweepers sensors' are shown in green and red areas indicate locations where potential evaders may still be located. (a) - Evader region status and sweepers' locations prior to the last inward advancement. (b) - Evader region status and sweepers' locations prior to the linear sweep.}
\label{Fig3Label}
\end{figure}

\newtheorem{thm}{Theorem}
\begin{thm}
When defining $\alpha= \frac{{R_0}}{r}$, if $\Delta V$ satisfies that,
\begin{equation}
\Delta V \ge \frac{{ - 4\pi {V_T}\alpha  + \pi {V_T} + {V_T}\sqrt {{\pi ^2} + 8\pi n} }}{{2n}}
\label{e29}
\end{equation}
then the evader region will be completely cleaned by $n$ sweepers that employ the linear scan after $N_n +1$ iterations.
\end{thm}
\begin{proof}
During the previously mentioned movement the margin between the tip of the sensor in each direction to the evader region boundaries must satisfy,
\begin{equation}
\frac{{2r - {R_{last}}}}{{{V_T}}} > {T_{linear}}
\label{e30}
\end{equation}
in order to guarantee no escape. ${T_{linear}}$ denotes the time it takes the sweepers to clean the right section of the remaining evader region in addition to the time it takes them to scan from the rightmost point they got to until the leftmost point of the expansion. These times are respectively denoted as $t$ and $ \tilde t$. Therefore, ${T_{linear}}$ is given by ${T_{linear}} = \tilde t + t$. The evader region's rightmost point of expansion starts from the point $(R_{last},0)$ and spreads at a speed of $V_T$. Therefore, if the constraint in (\ref{e30}) is satisfied we can view the rightward and leftward linear sweeps as a one dimensional scan. This geometric constraint can be observed in Fig. $3$. Therefore, the time $t$ it takes the sweepers to clean the spread of potential evaders from the right section of the region can be calculated from, ${V_s}t = {R_{last}} + {V_T}t$. Therefore, $t$ is given by, $t = \frac{{{R_{last}}}}{{{V_s} - {V_T}}}$. $\tilde t$ is computed by calculating the time it takes the sweepers located at point $(tV_s,0)$ to change their scanning direction and perform a leftward scan to a point that spread at a speed of $V_T$ from the leftmost point in the evader region at the origin of the search, the point $(-R_{last},0)$, for a time given by ${\tilde t + t}$. We have that, $- {R_{last}} - {V_T}\left( {\tilde t + t} \right) = t{V_s} - {V_s}\tilde t$. Plugging in the value of $t$ yields $\tilde t = \frac{{2{V_s}{R_{last}}}}{{{{\left( {{V_s} - {V_T}} \right)}^2}}}$. $T_{linear}$ is therefore given by,
\begin{equation}
{T_{linear}} = t + \tilde t = \frac{{6\pi r{V_T}{V_s} - 2\pi r{V_T}^2}}{{n{V_s}{{\left( {{V_s} - {V_T}} \right)}^2}}}
\label{e31}
\end{equation}
Therefore, the total scan time until a complete cleaning of the evader region is given by $T_{total} = {T_{circular}} + {T_{in}} + {T_{linear}}$. For the one dimensional scan to be valid and ensure a non escape search and complete cleaning of the evader region (\ref{e30}) must be satisfied. This demand implies that for a given $\alpha$, the designer of the sweep process can infer which $\Delta V $ needs to be chosen in order to satisfy (\ref{e29}) and thus completely clean the evader region using the final linear sweeping motion. For a complete derivation see Appendix $B$.
\end{proof}
\begin{thm}
For a valid circular search process the total search time until a complete cleaning of the evader region is given by, $T = {T_{circular}} + {T_{in}} + {T_{linear}}$, or as,
\begin{equation}
\begin{array}{l}
T =  - \frac{{{R_0}}}{{{V_T}}} + \frac{{r\left( {{V_s} - {V_T}} \right)\left( {n\left( {{V_s} + {V_T}} \right) + 2\pi {V_T}{N_n}} \right)}}{{2\pi {V_T}^2{V_s}}} +\\ {\left( {1 + \frac{{2\pi {V_T}}}{{n\left( {{V_s} + {V_T}} \right)}}} \right)^{{N_n} - 1}}\left( {\frac{{2\pi {R_0}{V_T} - rn\left( {{V_s} - {V_T}} \right)}}{{{V_s}}}} \right)\\\left( {\frac{1}{{n\left( {{V_s} + {V_T}} \right)}} + \frac{{{V_s}}}{{2\pi {V_T}^2}} + \frac{1}{{2\pi {V_T}}} + \frac{1}{{n{V_T}}}} \right) + \frac{{2\pi r}}{{n{V_s}}}\\
 + \frac{{6\pi r{V_T}{V_s} - 2\pi r{V_T}^2}}{{n{V_s}{{\left( {{V_s} - {V_T}} \right)}^2}}}
\end{array}
\label{e35}
\end{equation}
\end{thm}

\begin{figure}[htb!]
\noindent \centering{}\includegraphics[width=2.8in,height=2.3in]{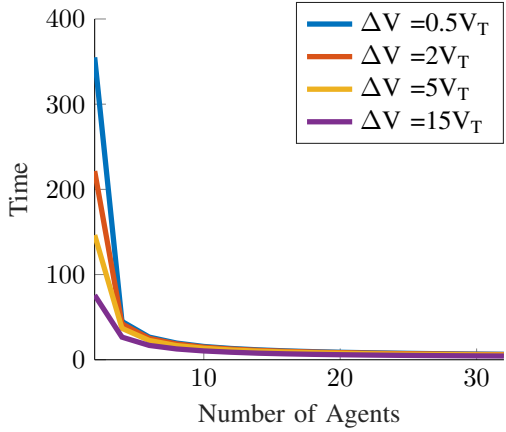} \caption{Total search times until complete cleaning of the evader region. We simulated sweep processes with an even number of agents, ranging from $2$ to $32$ agents, that employ the multi-agent same-direction circular sweep processes. We show results obtained for different values of speeds above the circular critical speed. The chosen values of the parameters are $r=10$, $V_T = 1$ and $R_0 = 100$.}
\label{Fig4Label}
\end{figure}

\section{Same-direction Spiral Sweep}
\subsection{Spiral Sweep Time Calculation}
Since our aim is to provide a sweep process that improves the same-direction circular sweep process we would like the sweepers to employ a more efficient motion throughout the sweep process. Therefore, we wish that throughout the motion of the sweepers, their sensors' footprint will maximally overlap the evader region. This can be obtained by a spiral scan, where the sweepers' sensors track the expanding evader region's wavefront, while preserving its shape to be as close as possible to a circle. An illustration of the initial placement of $2$ sweepers that employ the same-direction spiral sweep process is presented in Fig. $1 (b)$. The sweepers start with a sensor length of $2r$ inside the evader region. If the sweeper agents' speed is above the scenario's critical speed, the sweepers reduce the evader region's area after completing a traversal around the region. Each sweeper begins its spiral traversal with the tip of its sensor that is furthest from the center of the evader region, in a position that is tangent to the boundary of the evader region. In order to keep their sensors tangent to the evader region, the sweepers must travel at angle $\phi$ to the normal of the evader region. $\phi$ is calculated from $\sin \phi  = \frac{{{V_T}}}{{{V_s}}}$ This method of traveling at angle $\phi$ preserves the evader region's circular shape and is described thoroughly in \cite{francos2021tro}.

Contrary to the pincer-based strategy where each sweeper travels only an angle of $\frac{2\pi}{n}$ at each sweep iteration, in same-direction sweeps, each sweeper travels a larger angle than $\frac{2\pi}{n}$ at each iteration around the evader region in order to detect all escaping smart evaders. The additional angle, denoted by $\beta$, needs to be traversed in order to detect all evaders that may have spread from the "most dangerous points" at the beginning of each sweep. Such points are adjacent to the starting locations of every sweeper. The angle $\beta$ depends on the radius of the circle that bounds the evader region. After a sweeper traverses the additional angle $\beta$, the evader region's boundary is due to spread from points that resided at the lower tips of the sensors. When the tips of the sensors leave these points, evaders may spread from them in all directions at a speed of $V_T$. The time it takes a sweeper to travel an angle of $\frac{2\pi}{n} + {\beta _0}$, where ${\beta _0}$ is shown in Fig. $2 (b)$ is given by, 

\begin{equation}
{t_{\frac{2\pi}{n}  + {\beta _0}}} = \frac{{\left( {{R_0} - r} \right)\left( {{e^{\frac{{\left( {\frac{2\pi}{n}  + {\beta _0}} \right){V_T}}}{{\sqrt {{V_s}^2 - {V_T}^2} }}}} - 1} \right)}}{{{V_T}}}
\label{e309}
\end{equation}
The method of deriving this equation follows similar steps as in section $\text{V}$ of \cite{francos2021tro}. The subscript $0$ in ${\beta _0}$ denotes the iteration or cycle number, indicating that the value of $\beta$ changes as the sweep process progresses. After a sweeper completes a traversal of $\frac{2\pi}{n}  + {\beta _0}$ around the evader region it moves towards the center of the evader region. During this motion its lower tip points to the center of the region. $\beta _0$ is given by,
\begin{equation}
\sin {\beta _0} = \frac{{{V_T}{t_{\frac{2\pi}{n}  + {\beta _0}}}}}{{{R_0} - 2r}}
\label{e310}
\end{equation}
After a sweeper traverses $\frac{2\pi}{n} + \beta _0$ around the evader region the evader region's boundary is due to evaders that originated from the next "most dangerous points". The critical speed that satisfies the confinement task is computed numerically using the Newton method. When the sweepers travel towards the center of the evader region after completing a spiral sweep they have to meet the evader wavefront travelling outwards the region with a speed of $V_T$ at the previous radius $R_0$. Therefore, the expansion of the evader region after the first sweep at time ${t_{\frac{2\pi}{n}  + {\beta _0}}}$, has to satisfy that,
\begin{equation}
{V_T}{t_{\frac{2\pi}{n}  + {\beta _0}}} \leq \frac{{2r{V_s}}}{{{V_s} + {V_T}}}
\label{e330}
\end{equation}
The critical speed is obtained when we have equality in (\ref{e330}). In order to calculate $\beta_0$ that is obtained when the sweepers move at the critical speed, the expression of ${V_T}{t_{\frac{2\pi}{n}  + {\beta _0}}}$ in (\ref{e330}) is substituted with its equivalent expression in (\ref{e310}). Hence $\beta_0$ is, 
\begin{equation}
{\beta _0} = \arcsin \left( {\frac{{2r{V_s}}}{{\left( {{V_s} + {V_T}} \right)\left( {{R_0} - 2r} \right)}}} \right)
\label{e311}
\end{equation}
Substituting the expression for ${t_{\frac{2\pi}{n}  + {\beta _0}}}$, yields
\begin{equation}
\left( {{R_0} - r} \right)\left( {{e^{\frac{{\left( {\frac{2\pi}{n}  + {\beta _0}} \right){V_T}}}{{\sqrt {{V_s}^2 - {V_T}^2} }}}} - 1} \right) = \frac{{2r{V_s}}}{{{V_s} + {V_T}}}
\label{e317}
\end{equation}
In order to solve for $V_s$ we write (\ref{e317}) as,
\begin{equation}
F\left( {{V_s}} \right) = \frac{{2r{V_s}}}{{{V_s} + {V_T}}} - \left( {{R_0} - r} \right)\left( {{e^{\frac{{\left( {\frac{2\pi}{n}  + {\beta _0}} \right){V_T}}}{{\sqrt {{V_s}^2 - {V_T}^2} }}}} - 1} \right)
\label{e318}
\end{equation}
From (\ref{e318}) we find $V_s$ using the Newton iterative root finding method whose equation is given by,
\begin{equation}
{V_{{s_{n + 1}}}} = {V_{{s_n}}} - \frac{{F\left( {{V_{{s_n}}}} \right)}}{{\frac{{\partial F\left( {{V_{{s_n}}}} \right)}}{{\partial {V_{{s_n}}}}}}}
\label{e321}
\end{equation}
We choose as our initial estimate the lower bound on the sweeper speed (proved in \cite{francos2021tro}) given by ${V_{{s_0}}} = \frac{{\pi {R_0}{V_T}}}{nr}= V_{LB}$.
By using the described iterative convergence, we obtain a solution for $V_s$, which is the same-direction spiral sweep's critical speed. We denote this speed as ${V_{c_{spiral_{same}}}}$. After converging to a solution we obtain a result that is only slightly larger than the lower bound on the sweeper speed, $V_{LB}$.



Let us denote by $\Delta V >0$ the addition to the sweeper's speed above the critical speed. The speed is therefore given by, $V_s = {V_{c_{spiral_{same}}}} + \Delta V$. If a sweeper moves with a speed greater than the critical speed, after each spiral sweep it can advance inwards towards the center of the evader region and sweep around an evader region that is bounded by a circle with a smaller radius. The total search time until the evader region is bounded by a circle with a radius that is less than or equal to $2r$ is given by the sum of the total spiral sweep times and the times of the inward advances. Namely,
\begin{equation}
T = {T_{in}} + {T_{spiral}}
\label{e328}
\end{equation}
After each iteration, the sweepers move inwards towards the center of the evader region and the radius of the circle that bounds the region decreases. Consequently, the angle $\beta_i$ after which the sweepers move inwards changes as well. Therefore, after each sweep $\beta_i$ is calculated with respect to the new radius of the circle that bounds the evader region,
\begin{equation}
{\beta _i} = \arcsin \left( {\frac{{2r{V_s}}}{{\left( {{V_s} + {V_T}} \right)\left( {{R_i} - 2r} \right)}}} \right)
\label{e316}
\end{equation}
The time it takes to complete a spiral sweep of $\frac{2\pi}{n} +\beta_i$ around a region bounded by a circle of radius $R_i$ is given by, 
\begin{equation}
{T_{spira{l_i}}} = \frac{{\left( {{R_i} - r} \right)\left( {{e^{\frac{{\left( {\frac{2\pi}{n}  + {\beta _i}} \right){V_T}}}{{\sqrt {{V_s}^2 - {V_T}^2} }}}} - 1} \right)}}{{{V_T}}}
\label{e322}
\end{equation}
Denote the distance an agent can advance towards the center of the evader region by ${\delta _i}(\Delta V)$. In the term ${\delta _i}(\Delta V)$, $\Delta V$ denotes the increase in the agent speed relative to the critical speed, and $i$ denotes the number of sweep iterations the sweepers perform around the evader region, where $i$ starts from sweep number $0$. This results in a new evader region bounded by a circle with a radius of ${R_{i + 1}} = {R_i} - {\delta _i}(\Delta V)$. We have that,
\begin{equation}
{\delta _i}(\Delta V) = 2r - {V_T}{T_{spiral{i}}} \hspace{1mm}, \hspace{1mm} 0  \le {\delta _i}(\Delta V) \le 2r
\label{e324}
\end{equation}
As a function of the iteration number, we have that the distance a sweeper can advance inwards after completing an iteration is given by, 
\begin{equation}
{\delta _i}(\Delta V) = 2r - \left( {{R_i} - r} \right)\left( {{e^{\frac{{\left( {\frac{2\pi}{n}  + {\beta _i}} \right){V_T}}}{{\sqrt {{V_s}^2 - {V_T}^2} }}}} - 1} \right) 
\label{e325}
\end{equation}
After the sweeper completes a cycle, it moves inwards towards the center of the evader region in a way that the inner tip of its sensor points towards the center of the evader region with a speed of $V_s$, until it reaches the position where it starts its next sweep at the moment it meets the evader region's expanding wavefront. During the inwards advancements no cleaning is performed, while the evader region continues to spread. The time it takes the sweeper to move inwards until its entire sensor is over the evader region depends on the relative speed between the sweeper's inwards entry speed and the evader region outwards expansion speed and is given in (\ref{e329}). As the sweepers progress towards the center of the evader region, the evader region continuous to expand. Therefore, sweepers can only advance by a smaller distance than ${\delta _i}(\Delta V)$, denoted by ${\delta _{{i_{eff}}}}(\Delta V)$, which depends on the ratio between the speed in which the sweeper progresses towards the center of the region and the sum of velocities of sweeper and evader region spread. ${\delta _{{i_{eff}}}}(\Delta V)$ is the actual distance the sweeper moves at each iteration in order to meet the wavefront of the evader region when its entire sensor overlaps the evader region. Therefore, the distance sweepers can advance inwards after completing an iteration is given by,
\begin{equation}
{\delta _{{i_{eff}}}}(\Delta V) = {\delta _i}(\Delta V)\left( {\frac{{{V_s}}}{{{V_s} + {V_T}}}} \right)
\label{e326}
\end{equation}
The evader region is therefore bounded by a circle with a radius given by,
\begin{equation}
{R_{i + 1}} = {R_i} - {\delta _i}(\Delta V)\left( {\frac{{{V_s}}}{{{V_s} + {V_T}}}} \right)
\label{e327}
\end{equation}
This process continues until the evader region is bounded by a circle with a radius that is smaller than $2r$. We denote this radius as $R_N$. Once the evader region is contained inside a circular domain with a radius of $2r < R_i < 4r$,  $\beta_i$ is,
\begin{equation}
{\beta _i} = \arcsin \left( {\frac{{\left( {{R_i} - 2r} \right){V_s}}}{{\left( {{V_s} + {V_T}} \right)\left( {{R_i} - 2r} \right)}}} \right)
\label{e312}
\end{equation}
The inwards advancement time depends on the iteration number. It is denoted by $T_{i{n_i}}$ and its expression is given by,
\begin{equation}
{T_{i{n_i}}} = \frac{{{\delta _{{i_{eff}}}}(\Delta V)}}{{{V_s}}} = \frac{{2r - \left( {{R_i} - r} \right)\left( {{e^{\frac{{\left( {\frac{2\pi}{n}  + {\beta _i}} \right){V_T}}}{{\sqrt {{V_s}^2 - {V_T}^2} }}}} - 1} \right)}}{{{V_s} + {V_T}}}
\label{e329}
\end{equation}

During the inward advancements only the tip of the sensor, that has zero width, is inserted into the evader region. Therefore, no evaders are detected until the sweeper completes its inward advance and starts sweeping again. This search methodology continues until the evader region is bounded by a circle with a radius that is less than or equal to $2r$.
\subsection{Same-direction Spiral Sweep End-game}
In order to entirely clean the evader region, the sweepers need to change the scanning method when the evader region is bounded by a circle of radius $2r$, due to the same consideration that are described in the end-game of the same-direction circular sweep process. The depiction of the scenario at the beginning of the end-game is shown in Fig. $5$.
\begin{figure}
\noindent \centering{}\includegraphics[width=2.52in,height =2.52in]{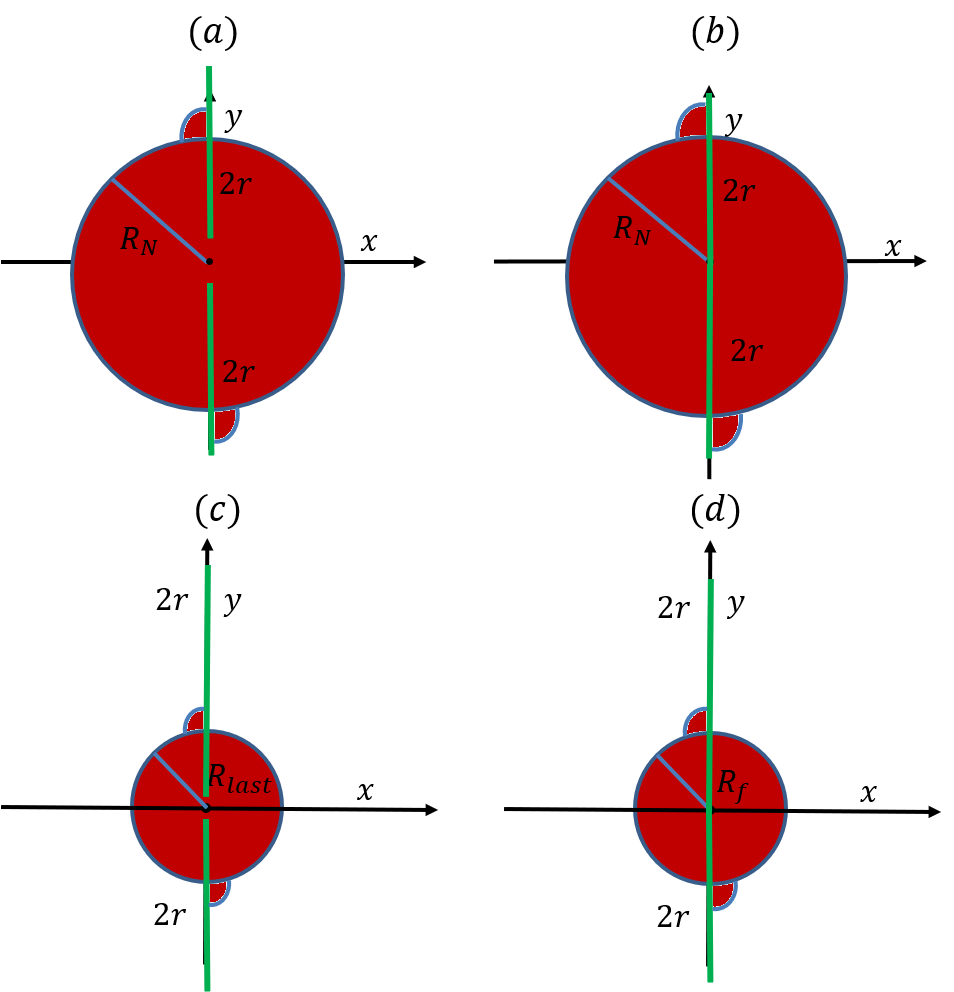} \caption{ Different stages of the end-game. The sweepers' sensors are shown in green and the evader region is shown in red. (a) - Depiction of the state at the beginning of the end-game. (b) - Depiction of the scenario after the sweepers move toward the center of the evader region and place the tips of their sensors at the center of the region. (c) - Depiction of the scenario after the last spiral sweep. (d) - Depiction of the scenario prior to the linear sweep.}
\label{Fig5Label}
\end{figure}
In the last inward advancement, the sweepers place the tips of their sensors at the center of the evader region. The time it takes the sweepers to complete this movement is given by ${T_e} = \frac{{{R_N}}}{{{V_s}}}$.    

Following the last inward advancement, the sweepers perform the last spiral sweep when the tips of their sensors are placed at the center of the evader region as can be seen in Fig. $5(b)$. In order to apply a linear sweeping movement the last spiral sweep has to reduce the evader region to be bounded by a circle of radius less than $2r$. Following the last inward advancement, the sweepers performs an additional spiral sweep when the center of each sweeper's sensor is as at a distance of $r$ from the center of the evader region. The time it takes to complete this sweep is denoted by $T_l$ and is given by,
\begin{equation}
{T_l} = \frac{{r\left( {{e^{\frac{{2\pi {V_T}}}{{n \sqrt {{V_s}^2 - {V_T}^2} }}}} - 1} \right)}}{{{V_T}}}
\label{e331}
\end{equation}
The depiction after the sweepers complete this last spiral movement can be seen in Fig. $5(c)$. During the last spiral sweep, the evader region spreads from its center point to a circle with a radius of, 
\begin{equation}
{R_{last}} = {T_l}{V_T} = r\left( {{e^{\frac{{2\pi {V_T}}}{{n\sqrt {{V_s}^2 - {V_T}^2} }}}} - 1} \right)
\label{e332}
\end{equation}
In order for a linear scan to be applicable ${R_{last}}$ has to be smaller than $2r$. This leads to a requirement on the sweeper's speed developed in Appendix $C$. Following this last spiral sweep, two sweepers place the tips of their sensors at the center of the evader region, advancing a distance of ${R_{in}} = 2r - {R_{last}}$. The time it takes the sweepers to complete this motion is given by ${T_{f}} = \frac{ 2r - {R_{last}}}{{{V_s}}}$. During this time the evaders spread to a circle of radius $R_f$ around the center of the region given by ${R_f} = {T_{f}}{V_T} + {R_{last}}$.

The depiction of the scenario at this time instance is shown in Fig. $5(d)$. Following this movement, the sweepers perform a linear motion and complete the search process. During the previously mentioned movement the margin between the edge of the sensors in each direction to the evader region boundaries must satisfy,
\begin{equation}
\frac{{2r - {R_{f}}}}{{{V_T}}} > {T_{linear}}
\label{e208}
\end{equation}
In order to guarantee no evader escapes undetected. The development of the linear sweep time follows the same steps as in the end-game section of the same-direction circular sweep when replacing $R_{f}$ instead of $R_{last}$. Therefore, the total scan time until a complete cleaning of the evader region is given by,

\begin{equation}
T_{total} = T_{spiral} + T_{in}+ T_e + T_l + T_f  + T_{linear}
\label{e223}
\end{equation}
For the one dimensional scan to be valid and ensure a non escape search and complete cleaning of the evader region, (\ref{e208}) must be satisfied. This demand results in the following requirement on ${V_s}$,

\begin{equation}
{V_s} \ge \frac{{2r{V_T} + {V_T}{R_f}}}{{2r - {R_f}}}
\label{e308}    
\end{equation}
For proof see Appendix $D$.

\begin{figure}[htb!]
\noindent \centering{}\includegraphics[width=2.8in,height=2.3in]{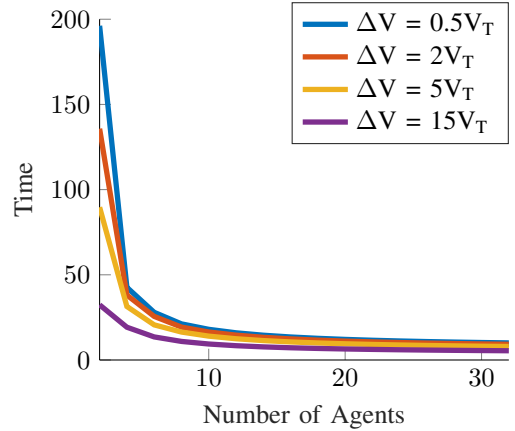} \caption{Total search times until complete cleaning of the evader region. We simulated sweep processes with an even number of agents, ranging from $2$ to $32$ agents, that employ the same-direction spiral sweep processes. We show results obtained for different values of speeds above the spiral critical speed. The chosen values of the parameters are $r=10$, $V_T = 1$ and $R_0 = 100$.}
\label{Fig6Label}
\end{figure}

\section{Comparative Analysis of Pincer Movement Search Strategies and Same-direction Sweeps}

The purpose of this section is to compare between the obtained results for the circular and spiral same-direction sweep processes that were developed in the previous sections and compare them to pincer sweep processes developed in \cite{francos2021tro}. In order to make a fair comparison between the total sweep times of sweeper swarms that can perform both the circular and spiral sweep processes the number of sweepers and sweeper speed must be the same in each of the tested spiral and circular swarms. The critical speed required for sweepers that perform the same-direction circular or spiral sweep processes is higher than the minimal critical speed of their pincer sweep counterparts. Additionally, since pincer sweep processes require a smaller critical speed compared to same-direction sweep processes, in order to make a fair comparison all sweepers in the swarm move at speeds above the critical speed of $2$ sweepers that perform the same-direction circular sweep.

Fig. $7$ shows the complete search times of swarms performing circular pincer sweeps, when the swarm's sweepers move at speeds above the same-direction circular critical speed. Fig. $8$ shows the complete search times of swarms performing the circular same-direction protocol. In both figures the sweepers move at the same speeds above the same-direction circular critical speed of $2$ sweepers, since this speed is greater than the critical speed of search processes performed with more sweepers. We see that for all choices of speeds above the same-direction circular critical speed of $2$ sweepers, the ratio between the complete search times of swarms performing same-direction circular sweeps and swarms performing their circular pincer sweep counterparts is greater than $1$, implying that same-direction circular sweeps require more time in order to clean the entire evader region. Hence, from these results we can conclude that performing circular pincer sweeps is always better than performing same-direction circular sweeps.

Fig. $9$ shows the complete search times of swarms performing spiral pincer sweeps, when the swarm's sweepers move at speeds above the same-direction spiral critical speed. Fig. $10$ shows the complete search times of swarms performing the spiral same-direction protocol. In both figures the sweepers move at the same speeds above the same-direction spiral critical speed of $2$ sweepers, since this speed is greater than the critical speed of search processes performed with more sweepers.

We see that for all choices of speeds above the same-direction spiral critical speed of $2$ sweepers, the ratio between the complete search times of swarms performing same-direction spiral sweeps and swarms performing their spiral pincer sweep processes counterparts is greater than $1$, implying that same-direction spiral sweeps require more time in order to clean the entire evader region. This result is expected since as the number of sweepers increases, the gain in utilizing the cooperation between the sweeping pairs in pincer-based sweep processes decreases the sweeping time more significantly compared to sweepers that perform the same-direction spiral sweeps. This occurs since same-direction sweepers must sweep larger angular sections at each iteration in order to ensure no evader escapes, while in pincer-based spiral search strategies, sweeping these additional sections is unnecessary due to the complementary trajectories of the sweepers. 

Hence, from these results we can conclude that performing spiral pincer sweeps is always better than performing same-direction sweeps.

Fig. $11$ presents a comparison between critical speeds required to perform each sweep protocol. Requiring a higher critical speed implies that there are entire regions of operation where an evader region with a given radius could be cleaned by a sweeper swarm that performs the same-direction spiral sweep process but cannot be cleaned by a sweeper swarm that performs the same-direction circular sweep process. This also implies that swarms that perform pincer movement search strategies can sweep larger regions than their same-direction sweeps counterparts.  

Furthermore, results show that as the number of sweepers increases, circular pincer-based protocols require a smaller critical speed even when compared to spiral same-direction protocols. This result indicates that although implementing pincer-based circular search protocols requires sweepers with more basic capabilities compared to spiral protocols, the cooperation between the sweepers considerably improves the overall performance of the sweeper team.

\begin{figure}[htb!]
\noindent \centering{}\includegraphics[width=2.8in,height=2.3in]{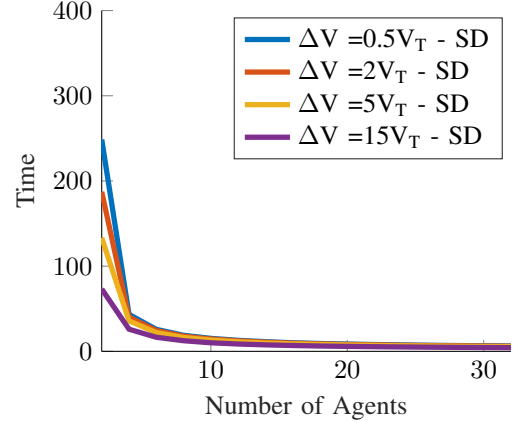} \caption{This plot shows the obtained sweep times for sweeper swarms performing the circular pincer sweep process that move at speeds above the critical speed of the circular same-direction process. The chosen values of the parameters are $r=10$, $V_T = 1$ and $R_0 = 100$.}
\label{Fig7Label}
\end{figure}

\begin{figure}[htb!]
\noindent \centering{}\includegraphics[width=2.8in,height=2.3in]{circular_sd_big.tikz} \caption{This plot shows the obtained sweep times for sweeper swarms performing the circular same-direction sweep process that move at speeds above the critical speed of the same-direction circular process. The chosen values of the parameters are $r=10$, $V_T = 1$ and $R_0 = 100$.}
\label{Fig8Label}
\end{figure}

\begin{figure}[htb!]
\noindent \centering{}\includegraphics[width=2.8in,height=2.3in]{spiral_pincer_big.tikz} \caption{This plot shows the obtained sweep times for sweeper swarms performing the spiral pincer sweep process that move at speeds above the critical speed of the spiral same-direction process. The chosen values of the parameters are $r=10$, $V_T = 1$ and $R_0 = 100$.}
\label{Fig9Label}
\end{figure}

\begin{figure}[htb!]
\noindent \centering{}\includegraphics[width=2.8in,height=2.3in]{spiral_same_direction_big.tikz} \caption{This plot shows the obtained sweep times for sweeper swarms performing the spiral pincer sweep process that move at speeds above the critical speed of the spiral same-direction process. The chosen values of the parameters are $r=10$, $V_T = 1$ and $R_0 = 100$.}
\label{Fig10Label}
\end{figure}


\begin{figure}[htb!]
\noindent \centering{}\includegraphics[width=2.8in,height=2.3in]{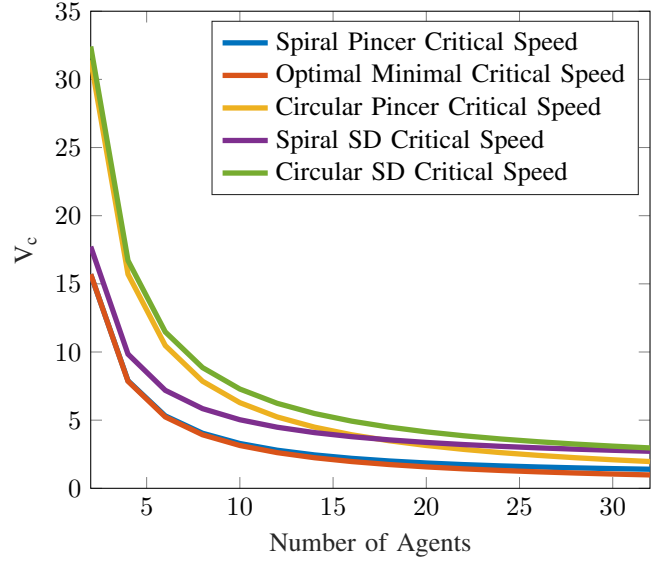} \caption{Critical speeds as a function of the number of sweepers. The number of sweepers is even, and ranges from $2$ to $32$ agents, that employ the spiral and circular pincer sweep process and the same-direction sweep protocols. The chosen values of the parameters are $r=10$, $V_T = 1$ and $R_0 = 100$.}
\label{Fig12Label}
\end{figure}

\section{Conclusions}
This work compares same-direction sweeps and pincer movement sweep protocols demonstrating and proving the superiority of the latter. We perform a quantitative comparison between pincer-based sweep protocols and same-direction sweep protocols for any number of even sweepers where the sensing capabilities and speeds of the swarms are equivalent. We prove that critical speeds for pincer based search methods are lower than their same-direction counterparts and therefore allow to sweep successfully larger regions. 

Afterwards, we provide a comparison between the different search methods in terms of completion times of the sweep processes and show that circular pincer-based approaches are always better than their same-direction counterparts. Furthermore, we show that pincer-based spiral sweep search times are shorter for all choices of search parameters compared to their same-direction counterparts, as well.

Hence, for all choices of search parameters and protocols the pincer-based protocols are best.

\appendices
\section{}
The inward advancement time at iteration $i$ is denoted by ${T_{i{n_i}}}$. It is given by,
\begin{equation}
{T_{i{n_i}}} = \frac{{{\delta _{{i_{eff}}}}(\Delta V)}}{{{V_s}}} = \frac{{rn\left( {{V_s} - {V_T}} \right) - 2\pi {R_i}{V_T}}}{{n{V_s}\left( {{V_s} + {V_T}} \right)}}
\label{e76}    
\end{equation}
The total advancement time until the evader region is bounded by a circle of with a radius that is less than or equal to r is denoted as $\widetilde{T}_{in}$. It is given by,
\begin{equation}
\widetilde{T}_{in} = \sum\limits_{i = 0}^{N - 2} {{T_{i{n_i}}}}  = \frac{{\left( {{N_n} - 1} \right)rn\left( {{V_s} - {V_T}} \right)}}{{n{V_s}\left( {{V_s} + {V_T}} \right)}} - \frac{{2\pi {V_T}\sum\limits_{i = 0}^{{N_n} - 2} {{R_i}} }}{{n{V_s}\left( {{V_s} + {V_T}} \right)}}
\label{e77}    
\end{equation}
We have that,
\begin{equation}
{R_{N_n - 2}} = \frac{{{c_1}}}{{1 - {c_2}}} + {c_2}^{N_n - 2}\left( {{R_0} - \frac{{{c_1}}}{{1 - {c_2}}}} \right)
\label{e78}    
\end{equation}
The sum of the radii is given by,
\begin{equation}
\sum\limits_{i = 0}^{N_n - 2} {{R_i}}  = \frac{{{R_0} - {c_2}{R_{N_n - 2}} + (N_n - 2){c_1}}}{{1 - {c_2}}}
\label{e79}    
\end{equation}
Where the coefficients $c_1$ and $c_2$ are given by,
\begin{equation}
{c_1} =  - \frac{{r\left( {{V_s} - {V_T}} \right)}}{{{V_s} + {V_T}}},{c_2} = 1 + \frac{{2\pi {V_T}}}{{n\left( {{V_s} + {V_T}} \right)}}
\label{e80}    
\end{equation}
Substitution of terms for the expression of ${R_{N_n - 2}}$ in (\ref{e78}) yields,
\begin{equation}
\begin{array}{l}
{R_{N_n - 2}} = \frac{{r\left( {{V_s} - {V_T}} \right)}}{{2\pi {V_T}}} + \\ {\left( {1 + \frac{{2\pi {V_T}}}{{{V_s} + {V_T}}}} \right)^{N_n - 2}}\left( {\frac{{2\pi {R_0}{V_T} - r\left( {{V_s} - {V_T}} \right)}}{{2\pi {V_T}}}} \right)
\end{array}
\label{e81}    
\end{equation}
Substitution of terms in (\ref{e79}) yields,
\begin{equation}
\begin{array}{l}
\sum\limits_{i = 0}^{{N_n} - 2} {{R_i}}  =  - \frac{{{R_0}n\left( {{V_s} + {V_T}} \right)}}{{2\pi {V_T}}} + \frac{{r{n^2}\left( {{V_s} - {V_T}} \right)\left( {{V_s} + {V_T}} \right)}}{{{{\left( {2\pi {V_T}} \right)}^2}}} + \\ \frac{{{N_n}nr\left( {{V_s} - {V_T}} \right)}}{{2\pi {V_T}}} - \frac{{rn\left( {{V_s} - {V_T}} \right)}}{{2\pi {V_T}}} + \\ {\left( {1 + \frac{{2\pi {V_T}}}{{n\left( {{V_s} + {V_T}} \right)}}} \right)^{{N_n} - 1}}\left( {\frac{{\left( {2\pi {R_0}{V_T} - rn\left( {{V_s} - {V_T}} \right)} \right)n\left( {{V_s} + {V_T}} \right)}}{{{{\left( {2\pi {V_T}} \right)}^2}}}} \right)
\end{array}
\label{e82}    
\end{equation}
We therefore obtain that,
\begin{equation}
\begin{array}{l}
\widetilde{T}_{in} = \frac{{{R_0}}}{{{V_s}}} - \frac{{rn\left( {{V_s} - {V_T}} \right)}}{{2\pi {V_T}{V_s}}}\\ - {\left( {1 + \frac{{2\pi {V_T}}}{{n\left( {{V_s} + {V_T}} \right)}}} \right)^{{N_n} - 1}}\left( {\frac{{2\pi {R_0}{V_T} - rn\left( {{V_s} - {V_T}} \right)}}{{2\pi {V_T}{V_s}}}} \right)
\end{array}
\label{e83}    
\end{equation}
The last inward advancement is given by,
\begin{equation}
{T_{_{in}last}} = \frac{{{R_{N_n}}}}{{{V_s}}}
\label{e84}    
\end{equation}
We have that,
\begin{equation}
{R_{N_n}} = \frac{{{c_1}}}{{1 - {c_2}}} + {c_2}^{N_n}\left( {{R_0} - \frac{{{c_1}}}{{1 - {c_2}}}} \right)
\label{e85}    
\end{equation}
Therefore,
\begin{equation}
\begin{array}{l}
{R_{{N_n}}} = \frac{{rn\left( {{V_s} - {V_T}} \right)}}{{2\pi {V_T}}} + \\ {\left( {1 + \frac{{2\pi {V_T}}}{{n\left( {{V_s} + {V_T}} \right)}}} \right)^{{N_n}}}\left( {\frac{{2\pi {R_0}{V_T} - rn\left( {{V_s} - {V_T}} \right)}}{{2\pi {V_T}}}} \right)
\end{array}
\label{e86}    
\end{equation}
Substitution of terms yields,
\begin{equation}
\begin{array}{l}
{T_{_{in}last}} = \frac{{rn\left( {{V_s} - {V_T}} \right)}}{{2\pi {V_T}{V_s}}} + \\  {\left( {1 + \frac{{2\pi {V_T}}}{{n\left( {{V_s} + {V_T}} \right)}}} \right)^{{N_n}}}\left( {\frac{{2\pi {R_0}{V_T} - rn\left( {{V_s} - {V_T}} \right)}}{{2\pi {V_T}{V_s}}}} \right)
\end{array}
\label{e87}    
\end{equation}
The total inward advancement times is therefore given by,
\begin{equation}
{T_{in}} = \frac{{{R_0}}}{{{V_s}}} + {\left( {1 + \frac{{2\pi {V_T}}}{{{V_s} + {V_T}}}} \right)^{N - 1}}\left( {\frac{{2\pi {R_0}{V_T} - r\left( {{V_s} - {V_T}} \right)}}{{{V_s}\left( {{V_s} + {V_T}} \right)}}} \right)
\label{e88}    
\end{equation}
The time it takes the sweepers to perform the circular sweeps before the evader region is bounded by a circle with a radius that is smaller or equal to $r$ is given by,
\begin{equation}
\widetilde{T}_{circular} = \frac{{{T_0} - {c_2}{T_{{N_n} - 1}} + \left( {{N_n} - 1} \right){c_3}}}{{1 - {c_2}}}
\label{e89}    
\end{equation}
Where the coefficient $c_3$ is given by,
\begin{equation}
{c_3} =  - \frac{{2\pi r\left( {{V_s} - {V_T}} \right)}}{{n{V_s}\left( {{V_s} + {V_T}} \right)}}
\label{e90}    
\end{equation}
The time it takes the sweepers to perform the first sweep is given by,
\begin{equation}
{T_0} =\frac{{2\pi {R_0}}}{{n{V_s}}}
\label{e91}    
\end{equation}
The time it takes the sweepers to perform the last circular sweep is given by,
\begin{equation}
\begin{array}{l}
{T_{{N_n} - 1}} = \frac{{r\left( {{V_s} - {V_T}} \right)}}{{{V_s}{V_T}}} + \\ {\left( {1 + \frac{{2\pi {V_T}}}{{n\left( {{V_s} + {V_T}} \right)}}} \right)^{{N_n} - 1}}\left( {\frac{{2\pi {R_0}{V_T} - rn\left( {{V_s} - {V_T}} \right)}}{{n{V_s}{V_T}}}} \right)
\end{array}
\label{e92}    
\end{equation}
Therefore $\widetilde{T}_{circular}$ is given by,
\begin{equation}
\begin{array}{l}
\widetilde{T}_{circular} = - \frac{{{R_0}\left( {{V_s} + {V_T}} \right)}}{{{V_s}{V_T}}} + \frac{{r\left( {{V_s} - {V_T}} \right)\left( {n\left( {{V_s} + {V_T}} \right) + 2\pi {V_T}{N_n}} \right)}}{{2\pi {V_T}^2{V_s}}} + \\ {\left( {1 + \frac{{2\pi {V_T}}}{{n\left( {{V_s} + {V_T}} \right)}}} \right)^{{N_n}}}\left( {{V_s} + {V_T}} \right)\left( {\frac{{2\pi {R_0}{V_T} - rn\left( {{V_s} - {V_T}} \right)}}{{2\pi {V_T}^2{V_s}}}} \right)
\end{array}
\label{e93}    
\end{equation}
The last circular sweep occurs when the lowest tips of the sweepers' sensors are located at the center of the evader region. It is given by,
\begin{equation}
{T_{last}} = \frac{{2\pi r}}{{{n V_s}}}
\label{e94}    
\end{equation}
Therefore the total circular traversal times are given by,
\begin{equation}
\begin{array}{l}
{T_{circular}} =  - \frac{{{R_0}\left( {{V_s} + {V_T}} \right)}}{{{V_s}{V_T}}} + \frac{{r\left( {{V_s} - {V_T}} \right)\left( {n\left( {{V_s} + {V_T}} \right) + 2\pi {V_T}{N_n}} \right)}}{{2\pi {V_T}^2{V_s}}} + \\ {\left( {1 + \frac{{2\pi {V_T}}}{{n\left( {{V_s} + {V_T}} \right)}}} \right)^{{N_n}}}\left( {{V_s} + {V_T}} \right)\left( {\frac{{2\pi {R_0}{V_T} - rn\left( {{V_s} - {V_T}} \right)}}{{2\pi {V_T}^2{V_s}}}} \right) + \\ \frac{{2\pi r}}{{n{V_s}}}
\end{array}
\label{e95}    
\end{equation}

\section{}
For the one dimensional scan to be valid and ensure a non escape search and complete cleaning of the evader region (\ref{e30}) must be satisfied. This demand implies that,

\begin{equation}
\frac{{2r - {R_{last}}}}{{{V_T}}} > \frac{{{R_{last}}\left( {3{V_s} - {V_T}} \right)}}{{{{\left( {{V_s} - {V_T}} \right)}^2}}}
\label{e32}
\end{equation}
By rearranging terms, (\ref{e32}) can be written as,
\begin{equation}
2r{\left( {{V_s} - {V_T}} \right)^2} > {R_{last}}{V_s}\left( {{V_s} + {V_T}} \right)
\label{e33}
\end{equation}
By substitution of $R_0$ with $\alpha r$ where $\alpha  > 1$ and by substituting the terms for $V_s$ and ${R_{last}}$, (\ref{e33}) resolves to a quadratic equation in $\Delta V$ that has only one positive root. This root is a monotonically decreasing function in $\alpha$, given by
\begin{equation}
\Delta V \ge \frac{{ - 4\pi {V_T}\alpha  + \pi {V_T} + {V_T}\sqrt {{\pi ^2} + 8\pi n} }}{{2n}}
\label{e34}
\end{equation}

We have that,
\begin{equation}
2r{\left( {{V_s} - {V_T}} \right)^2} > {R_{last}}{V_s}\left( {{V_s} + {V_T}} \right)
\label{e96}    
\end{equation}
And,
\begin{equation}
{V_s} = \frac{{2\pi {R_0}{V_T}}}{{rn}} + {V_T} + \Delta V \hspace{1mm} {R_{last}} = \frac{{2\pi r{V_T}}}{{n{V_s}}}
\label{e97}    
\end{equation}
Denoting $\alpha  = \frac{{{R_0}}}{r}$ and substituting the following terms in (\ref{e96}) yields,
\begin{equation}
2r{\left( {\frac{{2\pi {V_T}\alpha }}{n} + \Delta V} \right)^2} > \frac{{2\pi r{V_T}}}{n}\left( {\frac{{2\pi \alpha {V_T}}}{n} + 2{V_T} + \Delta V} \right)
\label{e98}    
\end{equation}
Rearranging terms yields a quadratic equation in $\Delta V$,
\begin{equation}
\begin{array}{{l}}
\Delta {V^2} + \Delta V\left( {\frac{{4\pi {V_T}\alpha  - \pi {V_T}}}{n}} \right)  + \frac{{4{\pi ^2}{V_T}^2{\alpha ^2} - 2{\pi ^2}\alpha {V_T}^2 - 2\pi n{V_T}^2}}{{{n^2}}} \\ > 0
\end{array}
\label{e99}    
\end{equation}
Equation (\ref{e99}) has a positive and a negative root. Since $\Delta V$ is non-negative we are interested only in the positive root. Therefore, in order to completely clean the evader region $\Delta V$ has to satisfy
\begin{equation}
\Delta V \ge \frac{{ - 4\pi {V_T}\alpha  + \pi {V_T} + {V_T}\sqrt {{\pi ^2} + 8\pi n} }}{{2n}}
\label{e100}    
\end{equation}

\section{}
In order for a linear scan to be applicable ${R_{last}}$ has to be smaller than $2r$. This leads to a requirement on the sweeper's velocity,
\begin{equation}
r\left( {{e^{\frac{{2\pi {V_T}}}{{n\sqrt {{V_s}^2 - {V_T}^2} }}}} - 1} \right) < 2r    
\label{e200}    
\end{equation}
Which resolves to,
\begin{equation}
{e^{\frac{{2\pi {V_T}}}{{n\sqrt {{V_s}^2 - {V_T}^2} }}}} < 2
\label{e201}    
\end{equation}
Yielding the requirement on the velocity,
\begin{equation}
{V_s} > {V_T}\sqrt {\frac{{4{\pi ^2}}}{{{{\left( {n\ln 2} \right)}^2}}} + 1} 
\label{e202}    
\end{equation}
Since we previously observed that the spiral critical velocity (for the considered spiral process) is close to the lower bound on the critical velocity, $V_{LB}$, we can check whether the condition in (\ref{e200}) is automatically satisfied. If the requirement on $V_s$ in (\ref{e202}) is less than $V_{LB}$, then since the sweepers move with a velocity above it, then (\ref{e200}) is always satisfied. Therefore if,
\begin{equation}
{V_T}\sqrt {\frac{{4{\pi ^2}}}{{{{\left( {n\ln 2} \right)}^2}}} + 1}  < \frac{{\pi {R_0}{V_T}}}{nr} = V_{LB}
\label{e203}    
\end{equation}
Or if the ratio $\frac{{{R_0}}}{r}$ satisfies that,
\begin{equation}
\sqrt {\frac{4}{{{{\left( {n\ln 2} \right)}^2}}} + \frac{1}{{{\pi ^2}}}}  < \frac{{{R_0}}}{r}
\label{e204}    
\end{equation}
Then ${R_{last}}$ is smaller than $2r$.
\section{}
For the one dimensional scan to be valid and ensure a non escape search and complete cleaning of the evader region, (\ref{e208}) must be satisfied. This demand implies that,
\begin{equation}
\frac{{2r - {R_{f}}}}{{{V_T}}} > \frac{{{R_{f}}\left( {3{V_s} - {V_T}} \right)}}{{{{\left( {{V_s} - {V_T}} \right)}^2}}}
\label{e333}
\end{equation}
By rearranging terms, (\ref{e333}) can be written as,
\begin{equation}
2r{\left( {{V_s} - {V_T}} \right)^2} > {R_{f}}{V_s}\left( {{V_s} + {V_T}} \right)
\label{e334}
\end{equation}
By substitution of $R_0$ with $\alpha r$ where $\alpha  > 1$ and by substituting the terms for $V_s$ and ${R_{f}}$, (\ref{e334}) resolves to a quadratic equation in $\Delta V$ that has only one positive root. This root is a monotonically decreasing function in $\alpha$, given by
\begin{equation}
{V_s}^2\left( {2r - {R_f}} \right) - {V_s}{V_T}\left( {4r + {R_f}} \right) + 2r{V_T}^2 > 0  
\label{e335}    
\end{equation}
The quadratic equation in (\ref{e335}) has $2$ positive roots. Therefore, in order for the one dimensional linear scan to be valid $V_s$ has to be greater than the largest positive root. Implying that,
\begin{equation}
{V_s} \ge \frac{{2r{V_T} + {V_T}{R_f}}}{{2r - {R_f}}}
\label{e336}    
\end{equation}

\ifCLASSOPTIONcaptionsoff
  \newpage
\fi
\bibliographystyle{IEEEtran}
\bibliography{pincer_same_direction_search_for_smart_evaderscomparison}
\end{document}